\documentclass{article}
\usepackage[bookmarks]{hyperref}
\usepackage{
	amsfonts,
	amsmath,
	amssymb,
	amsthm,
	array,
	caption,
	chngpage,
	color,
	float,
	graphicx,
	latexsym,
	mathrsfs,
	multicol,
	multirow,
	sectsty,
	subcaption,
	tikz,
	tikz-cd,
	verbatim,
	wrapfig
}
\usetikzlibrary{automata, positioning}
\usetikzlibrary{decorations.pathmorphing}
\tikzset{snake it/.style={decorate, decoration=snake}}
\allsectionsfont{\centering}
\newcommand{\pars}[1]{\left(#1\right)}

\newcommand{\abs}[1]{\lvert#1\rvert}
\newcommand{\mr}{\mathrm}
\newtheorem{thm}{Theorem}
\newtheorem{cor}[thm]{Corollary}
\newtheorem{lem}[thm]{Lemma}

\newtheorem{df}[thm]{Definition}
\newtheorem{rem}[thm]{Remark}
\newtheorem{exa}[thm]{Example}
\DeclareMathOperator*{\argmax}{arg\, max}
\DeclareMathOperator{\logit}{logit}

\begin{document}
	\title{
		Kolmogorov structure functions for automatic complexity
	}
	\author{
		Bj{\o}rn Kjos-Hanssen\footnote{
			This work was partially supported by
			a grant from the Simons Foundation (\#315188 to Bj\o rn Kjos-Hanssen).
			The author also acknowledges the support of the Institute for Mathematical Sciences of the National University of Singapore
			during the workshop on \emph{Algorithmic Randomness}, June 2--30, 2014.
		}
	}
	\maketitle{}
	\begin{abstract}
		For a finite word $w$
		we define and study the Kolmogorov structure function $h_w$ for nondeterministic automatic complexity.
		We prove upper bounds on $h_w$ that appear to be quite sharp, based on numerical evidence.
	\end{abstract}

	\section{Introduction}
		Shallit and Wang \cite{MR1897300} introduced automatic complexity as a computable alternative to Kolmogorov complexity.
		They considered deterministic automata, whereas Hyde and Kjos-Hanssen \cite{COCOON:14} studied the nondeterministic case,
		which in some ways behaves better.
		Unfortunately, even nondeterministic automatic complexity is somewhat inadequate.
		The string $00010000$ has maximal nondeterministic complexity, even though intuitively it is quite simple.
		One way to remedy this situation is to consider a structure function analogous to that for Kolmogorov complexity.

		The latter was introduced by Kolmogorov at a 1973 meeting in Tallinn and studied by
		Vereshchagin and Vit\'anyi \cite{MR2103496} and Staiger \cite{StaigerTCS}.

		\begin{figure}[h]
			\begin{tikzpicture}[shorten >=1pt,node distance=1.5cm,on grid,auto]
				\node[state,initial, accepting] (q_1)   {$q_1$}; 
				\node[state] (q_2)     [right=of q_1   ] {$q_2$}; 
				\node[state] (q_3)     [right=of q_2   ] {$q_3$}; 
				\node[state] (q_4)     [right=of q_3   ] {$q_4$};
				\node        (q_dots)  [right=of q_4   ] {$\ldots$};
				\node[state] (q_m)     [right=of q_dots] {$q_m$};
				\node[state] (q_{m+1}) [right=of q_m   ] {$q_{m+1}$}; 
				\path[->] 
					(q_1)     edge [bend left]  node           {$x_1$}     (q_2)
					(q_2)     edge [bend left]  node           {$x_2$}     (q_3)
					(q_3)     edge [bend left]  node           {$x_3$}     (q_4)
					(q_4)     edge [bend left]  node [pos=.45] {$x_4$}     (q_dots)
					(q_dots)  edge [bend left]  node [pos=.6]  {$x_{m-1}$} (q_m)
					(q_m)     edge [bend left]  node [pos=.56] {$x_m$}     (q_{m+1})
					(q_{m+1}) edge [loop above] node           {$x_{m+1}$} ()
					(q_{m+1}) edge [bend left]  node [pos=.45] {$x_{m+2}$} (q_m)
					(q_m)     edge [bend left]  node [pos=.4]  {$x_{m+3}$} (q_dots)
					(q_dots)  edge [bend left]  node [pos=.6]  {$x_{n-3}$} (q_4)
					(q_4)     edge [bend left]  node           {$x_{n-2}$} (q_3)
					(q_3)     edge [bend left]  node           {$x_{n-1}$} (q_2)
					(q_2)     edge [bend left]  node           {$x_n$}     (q_1);
			\end{tikzpicture}
			\caption{
				A nondeterministic finite automaton that only accepts one string
				$x= x_1x_2x_3x_4 \ldots x_n$ of length $n = 2m + 1$.
			}
			\label{fig1}
		\end{figure}

		The Kolmogorov complexity of a finite word $w$ is roughly speaking
		the length of the shortest description $w^*$ of $w$ in a fixed formal language.
		The description $w^*$ can be thought of as an optimally compressed version of $w$.
		Motivated by the non-computability of Kolmogorov complexity,
		Shallit and Wang studied a deterministic finite automaton analogue.
		\begin{df}[Shallit and Wang \cite{MR1897300}]
			The \emph{automatic complexity} of a finite binary string \(x=x_1\dots x_n\) is 
			the least number \(A_D(x)\) of states of a {deterministic finite automaton} \(M\) such that 
			\(x\) is the only string of length \(n\) in the language accepted by \(M\).
		\end{df}
		Hyde and Kjos-Hanssen \cite{COCOON:14} defined a nondeterministic analogue:
		\begin{df}\label{precise}
			The nondeterministic automatic complexity $A_N(w)$ of a word $w$ is the minimum number of states of an NFA $M$,
			having no $\epsilon$-transitions, accepting $w$
			such that there is only one accepting path in $M$ of length $\abs{w}$.
		\end{df}
		The minimum complexity $A_N(w)=1$ is only achieved by words of the form $a^n$ where $a$ is a single letter.
		\begin{df}\label{df:KayleighGraph}
			Let $n = 2m + 1$ be a positive odd number, $m\ge 0$.
			A finite automaton of the form given in Figure \ref{fig1} for some choice of symbols $x_1,\dots,x_n$ and states
			$q_1,\dots,q_{m+1}$
			is called a \emph{Kayleigh graph}\footnote{
				The terminology is a nod to the more famous Cayley graphs as well as to Kayleigh Hyde's first name.
			}.
		\end{df}
		\begin{thm}[Hyde \cite{Hyde}]\label{Hyde}
			The nondeterministic automatic complexity $A_N(x)$ of a string $x$ of length $n$ satisfies
			\[
				A_N(x) \le b(n):={\lfloor} n/2 {\rfloor} + 1\text{.}
			\]
		\end{thm}
		\begin{proof}
			If $n$ is odd, then a Kayleigh graph witnesses this inequality.
			If $n$ is even, a slight modification suffices, see \cite{Hyde}.
		\end{proof}	
		The structure function of a string $x$ is defined by
		$h_x(m)= \min \{ k:$ there is a $k$-state NFA $M$ which accepts at most $2^m$ strings of length $\abs{x}$ including $x\}$.
		In more detail:

		Let
		\[
			S_x = \{(q,m) \mid \exists \text{ $q$-state NFA $M$}, x\in L(M)\cap \Sigma^n, \abs{L(M)\cap\Sigma^n}\le b^m\}.
		\]
		Then $S_x$ has the upward closure property
		\[
			q\le q', m\le m', (q,m)\in S_x \quad\Longrightarrow\quad (q',m')\in S_x.
		\]
		From $S_x$ we can define the structure function $h_x$ and the dual structure function $h_x^*$.
		\begin{df}[Vereshchagin, personal communication, 2014, inspired by \cite{MR2103496}]
			In an alphabet $\Sigma$ containing $b$ symbols, we define 
			\[
				h^*_x(m) = \min\{k : (k,m)\in S_x \}\quad\text{and}
			\]
			\[
				h_x(k) = \min\{m: (k,m)\in S_x\}.
			\]
		\end{df}
		\begin{rem}
			On the one hand, $h$ mimics the structure function as defined by Kolmogorov.
			On the other hand, $h^*$ has a natural domain $[0,n]$ whereas the domain of $h$ is initially $[1,\infty)$,
			until some upper bound on the automatic complexity is proved, at which point it becomes $[1,\lfloor n/2\rfloor + 1]$.
			One often prefers that a function have a simple domain and a complicated range rather than the other way around, e.g., consider
			the case of the range of a computable function on $\mathbb N$ (which is only computably enumerable).
		\end{rem}
		\paragraph{History of the structure function.}
		Kolmogorov first introduced the structure function in a talk at
		The Third International Symposium on Information Theory, June 18--23, 1973, Tallinn, Estonia, Soviet Union.
		The meeting coincided with a Nixon/Brezhnev meeting in the U.S. Kolmogorov was born in 1903 hence 70 years old at the time.
		The results were not published until they appeared as an abstract of a talk for the Moscow Mathematical Society \cite{Moscow} in
		\emph{Uspekhi Mat.\ Nauk} in the
		Communications of the Moscow Mathematical Society, page 155 (in the Russian edition, not translated into English).
		The talk was given on April 16, 1974 and was entitled ``Complexity of algorithms and objective definition of randomness''.

	\section{Basic properties}
		\begin{df}\label{df:entropy}
			The entropy function $\mathcal H:[0,1]\rightarrow [0,1]$ is given by
			\[
				\mathcal H(p)=-p\log_2p-(1-p)\log_2(1-p).
			\]
		\end{df}
		\begin{rem}
			Throughout the paper, $\log$ (with no subscript) denotes either the natural logarithm $\ln=\log_e$, or $\log_b$ where the value of $b$ is immaterial.
		\end{rem}
		\begin{thm}\label{familiar2Andre}
			For $0\le k\le n$,
			\[
				\log_2 {n\choose k} = \mathcal H(k/n) n + O(\log n).
			\]
		\end{thm}
		\begin{proof}
			For $u\in\mathbb N$, let
			\[
				S_u = \sum_{k=2}^u \log k,\quad I_u = \int_1^u \log x\,dx,\quad\text{and}\quad J_u = \int_2^{u+1} \log x\,dx.
			\]
			Let
			\[
				\alpha_n = \log {n\choose k}
				= S_n - S_{k} - S_{n-k}.
			\]
			Note $I_u\le S_u\le J_u$ and
			\[
				J_u-I_u =
				\int_u^{u+1} \log x\,dx - \int_1^2\log x\,dx \le \log(u+1),
			\]
			Thus up to $O(\log n)$ error terms we have
			\[
				\alpha_n = \int^n_1 \log x\,dx - \int_1^{k}\log x\,dx - \int_1^{n-k}\log x\,dx
			\]
			\[
				= (n\log n - n) - (k \log(k) - k) - [(n-k)\log(n-k) - (n-k)]
			\]
			\[
				= n\log n - k \log(k) - (n-k)\log(n-k)
			\]
			\[
				= - k \log(k/n) - (n-k)\log\pars{1-\frac{k}{n}}
			\]
			and hence
			\[
				\log_2 {n\choose k}
				= - k \log_2(k/n) - (n-k)\log_2(1-k/n) = \mathcal H(k/n)\cdot n.
			\]
		\end{proof}
		\begin{figure}[h]
			\begin{tikzpicture}[shorten >=1pt,node distance=1.5cm,on grid,auto]
				\node[state,initial] (q_1)   {$q_1$}; 
				\node[state] (q_2)     [right=of q_1   ] {$q_2$}; 
				\node[state] (q_3)     [right=of q_2   ] {$q_3$}; 
				\node[state] (q_4)     [right=of q_3   ] {$q_4$};
				\node        (q_dots)  [right=of q_4   ] {$\ldots$};
				\node[state] (q_m)     [right=of q_dots] {$q_m$};
				\node[state, accepting] (q_{m+1}) [right=of q_m   ] {$q_{m+1}$}; 
				\path[->] 
					(q_1)     edge [bend left]  node           {0}     (q_2)
					(q_2)     edge [bend left]  node           {0}     (q_3)
					(q_3)     edge [bend left]  node           {0}     (q_4)
					(q_4)     edge [bend left]  node [pos=.45] {0}     (q_dots)
					(q_dots)  edge [bend left]  node [pos=.6]  {0} (q_m)
					(q_m)     edge [bend left]  node [pos=.56] {0}     (q_{m+1})
					(q_1)     edge [loop above] node           {1} ()
					(q_2)     edge [loop above] node           {1} ()
					(q_3)     edge [loop above] node           {1} ()
					(q_4)     edge [loop above] node           {1} ()
					(q_{m+1}) edge [loop above] node           {1} ()
					(q_m)     edge [loop above] node           {1} ();
			\end{tikzpicture}
			\caption{
				An automaton illustrating multi-run complexity for a string of length $n$ containing $m$ many 0s, and $n-m$ many 1s.
			}
			\label{figMultiRun}
		\end{figure}
		\begin{thm}\label{thm:LowerEntropy}
			Suppose the number of $0$s in the binary string $x$ is $p\cdot n$.
			Then
			\[
				h^*_x(\mathcal H(p)n)\le pn + O(\log n).
			\]
		\end{thm}
		\begin{proof}
			Consider an automaton $M$ as in Figure \ref{figMultiRun} that has $p\cdot n$ many states, and
			that has one left-to-right arrow labeled 0 for each 0, and a loop in place labeled 1 for each consecutive string of 1s.
			Since $M$ accepts exactly those strings that have $p\cdot n$ many 0s,
			the number of strings accepted by $M$ is ${n\choose p\cdot n}$.
			By Theorem \ref{familiar2Andre} this is $\le 2^k$ approximately when $\mathcal H(p)n\le k$, and we are done.
		\end{proof}
		\begin{exa}
			A string of the form $0^a1^{n-a}$ satisfies $h^*_x(\log_2 n) = 2$ whereas $h^*_x(0)$ may be $n/2$.
			For instance $0011$ has $h^*_x(2)=2$. On the other hand $h^*_x(1)=3$ which is why this string is more complicated than $0110$.
		\end{exa}
		\begin{figure}[h]
			\begin{tikzpicture}[shorten >=1pt,node distance=1.5cm,on grid,auto]
				\node[state,initial] (q_1)   {$q_1$}; 
				\node[state] (q_2)     [right=of q_1   ] {$q_2$}; 
				\node[state] (q_3)     [right=of q_2   ] {$q_3$}; 
				\node[state] (q_4)     [right=of q_3   ] {$q_4$};
				\node        (q_dots)  [right=of q_4   ] {$\ldots$};
				\node[state] (q_m)     [right=of q_dots] {$q_m$};
				\node[state, accepting] (q_{m+1}) [right=of q_m   ] {$q_{m+1}$}; 
				\path[->] 
					(q_1)     edge [bend left]  node           {$x_1$}     (q_2)
					(q_2)     edge [bend left]  node           {$x_2$}     (q_3)
					(q_3)     edge [bend left]  node           {$x_3$}     (q_4)
					(q_4)     edge [bend left]  node [pos=.45] {$x_4$}     (q_dots)
					(q_dots)  edge [bend left]  node [pos=.6]  {$x_{m-1}$} (q_m)
					(q_m)     edge [bend left]  node [pos=.56] {$x_m$}     (q_{m+1})
					(q_{m+1}) edge [loop above] node           {0} ()
					(q_{m+1}) edge [loop below] node           {1} ();
			\end{tikzpicture}
			\caption{
				An automaton illustrating the linear upper bound on the automatic structure function from Theorem \ref{bound}.
			}
			\label{figLinearBound}
		\end{figure}
		\begin{thm}\label{bound}
			For any $x$ of length $n$,
			\[
			 	1\le h^*_x(m)\le n-m+1\text{ for }0\le m\le n.
			\]
		\end{thm}
		\begin{proof}
			$1\le h^*_x(n-k)\le k+1$ because we can start out with a sequence of determined moves,
			after which we accept everything, as in Figure \ref{figLinearBound}.
		\end{proof}
		\section{Upper bounds on structure function for automatic complexity}
			\begin{df}
				The dual automatic structure function of a string $x$ of length $n$ is a function
				$h^*_x:[0,n]\rightarrow [0,\lfloor n/2\rfloor + 1]$.
				We define the \emph{asymptotic upper envelope} of $h^*$ by
				\[
					\widetilde{h^*}(a)
					= \limsup_{n\rightarrow\infty}\max_{\abs{x}=n}\frac{h^*_x([a\cdot n])}{n},\quad \widetilde{h^*}:[0,1]\rightarrow [0,1/2]
				\]
				where $[\cdot ]$ is the nearest integer function.
				Let
				\[
					\tilde h(p) = \limsup_{n\rightarrow\infty}\max_{\abs{x}=n}\frac{h_x([p\cdot n])}{n},\quad \tilde h:[0,1/2]\rightarrow [0,1].
				\]
			\end{df}
			\begin{thm}[Main Theorem]\label{best}
				Assume $x$ is a binary string, so the alphabet size $b=2$.
				The asymptotic upper envelope $\tilde h$ of the automatic structure functions $h_x$ satisfies
				\[
					\tilde h(p)\le \psi(p):=
					\begin{cases}
						\mathcal H(\frac12-p),	 & \frac{\sqrt 3}{4} \le p\le \frac12,\\
						2-\alpha p,				 & \frac1{\alpha-1} \le p\le \frac{\sqrt 3}4,\\
						1-p,					 & 0\le p\le \frac1{\alpha -1},
					\end{cases}
				\]
				where
				\[
					\alpha=\frac{4}{\sqrt 3}\pars{
						2-\mathcal H\pars{
							\frac12-\frac{\sqrt 3}4
						}
					} =\mathcal H'\pars{
						\frac12-\frac{\sqrt 3}4
					} \approx 3.79994,
				\]
				\[
					\alpha = 2\log_2(2+\sqrt{3}).
				\]
			\end{thm}
			As Theorem \ref{best} shows, the largest number of paths is obtained by
			going \emph{fairly} straight to the loop state;
			spending half the time looping and half the time meandering; and then finally
			going equally fairly straight to the start state.
			The optimal value of $r$ obtained shows that half of
			the time between first reaching the loop state and finally leaving the loop state should be spent looping.
			Figures \ref{entropy} and \ref{entropy2} show our upper bounds for the automatic structure function.
			
			\begin{figure}[H]
				\begin{center}
					\includegraphics[height=7cm]{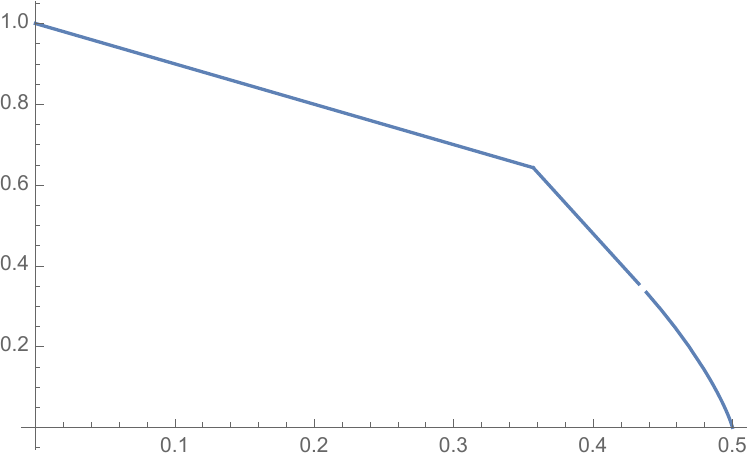}
				\end{center}
				\caption{Bounds for the automatic structure function for alphabet size $b=2$ when $\vec c=(1,-1,0)$; see Theorem \ref{best}. Figure produced using Mathematica with
				$y=\min\left(2-{2\log_2(2+\sqrt{3})}x, 1-x\right)$ for $0\le x\le\frac{\sqrt{3}}{4}$, and $y={-(\frac12-x)\log_2(\frac12-x)-(\frac12+x)\log_2(\frac12+x)}$ for $\frac{\sqrt{3}}{4}\le x$. 
				}
				\label{entropy}
			\end{figure}

			\begin{figure}[H]
				\begin{center}
					\includegraphics[height=7cm]{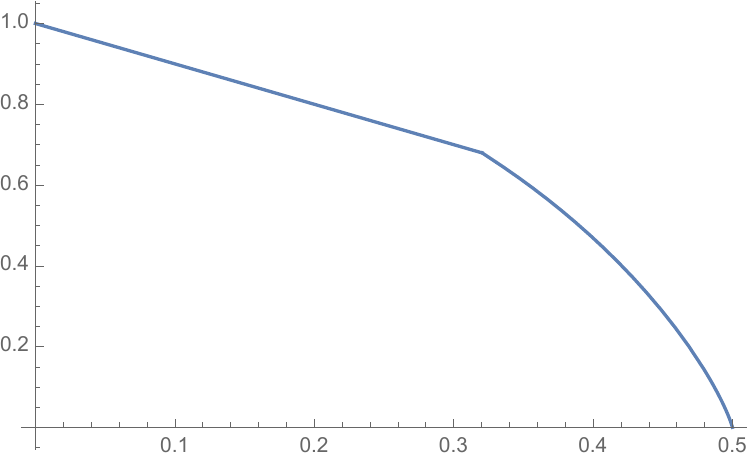}
				\end{center}
				\caption{Bounds for the automatic structure function for alphabet size $b=2$; see Theorem \ref{best}.
					When $\vec c=\frac12(1,-1,1)$, the entropy function is used on $[0.33,0.5]$, the tangent line on $[0.3205,0.330]$, and
					$y=1-x$ on $[0,0.3205]$.
					Figure produced using Mathematica with
					$y=\min(\sqrt{2}-2.29244x, 1-x)$ for $0\le x<0.330$, and
					$y={-(\frac12-x)\log_2(\frac12-x)-(\frac12+x)\log_2(\frac12+x)}$ for $0.330\le x$.
				}
				\label{entropy2}
			\end{figure}

			\section{Proof of Theorem \ref{best}}
				Consider a path of length $n$ through a Kayleigh graph with $q=pn$ many states.
				Let $t_1$ be the time spent before reaching the loop state for the first time.
				Let $t_2$ be the time spent after leaving the loop state for the last time.
				Let $s$ be the number of self-loops taken by the path.
				Let us say that \emph{meandering} is
				the process of leaving the loop state after having gone through a loop, and before again going through a loop.
				For fixed $p$ let
				\[
					\gamma(t_1,t_2,s,n) = {t_1\choose \frac{t_1-pn}{2}}
					{t_2\choose \frac{t_2-pn}{2}}
					{\chi(n,t,s) \choose s}
					{\chi(n,t,s) - s \choose \frac{\chi(n,t,s)-s}{2}} b^s
				\]
				where $\chi(n,t,s)=n-t$ and $t=t_1+t_2$.
				(By Lemma \ref{largeNo1}, we can also let $\chi(n,t,s)=(n-t+s)/2$, since the number of non-loops between loops must be even. 
				This gives a better upper bound.)

				Then the number of such paths is
				\begin{equation}
					N\le\sum_s\sum_{t_1}\sum_{t_2}\gamma(t_1,t_2,s,n)\label{1}
				\end{equation}
				since half of the meandering times must be backtrack times.
				\begin{lem}\label{largeNo1}
					Suppose $0\le k\le n$ with $n-k$ even.
					The number of $k$-element subsets of $\{1,\dots,n\}$ where
					the number of other elements between consecutive elements in the subset is always even is
					\[
						{(n-k)/2+k\choose k}.
					\]
				\end{lem}
				\begin{proof}
					The other elements come in pairs hence by merging the pair to one there are only $(n-k)/2$ of them.
				\end{proof}
				For instance, if $n=6$ and $k=2$, we get ${4\choose 2}=6$.
				Since
				\[
					   \limsup_{n\rightarrow\infty}\frac{\log_b\sum_1^n a_i}{n}
					\le\limsup_{n\rightarrow\infty}\frac{\log_b (n\cdot\max a_i)}{n}
					 = \limsup_{n\rightarrow\infty} \frac{\log_b\max a_i}{n},
				\]
				the sums can be replaced by maxima, i.e.,
				\[
					\limsup_{n\rightarrow\infty}\frac{\log_b N}{n}
					\le \limsup_{n\rightarrow\infty}\frac{\log_b \gamma(t_1,t_2,s,n)}{n},\quad (t_1,t_2,s)\in\argmax_{(t_1,t_2,s)}\gamma(t_1,t_2,s,n).
				\]
				By Theorem \ref{familiar2Andre},
				\[
					\limsup_{n\rightarrow\infty}\frac{\gamma(t_1,t_2,s,n)}{n}\le\limsup_{n\rightarrow\infty}\frac{\delta(t_1,t_2,s,n)}n
				\]
				where $\delta$ is
				\begin{eqnarray*}
					\sum_{i=1}^2 t_i \mathcal H_b\left(\frac12-\frac{pn}{2t_i}\right)
					+\left(\chi(n,t,s)\right)\mathcal H_b\left(\frac{s}{\chi(n,t,s)}\right)
					+&\left(\chi(n,t,s)-s\right)\mathcal H_b\left(\frac12\right)+s\\
					= \sum_{i=1}^2 t_i \mathcal H_b\left(\frac12-\frac{pn}{2t_i}\right)
					+\left(\chi(n,t,s)\right)\mathcal H_b\left(\frac{s}{\chi(n,t,s)}\right)
					+&\chi(n,t,s)\log_b 2 + (1-\log_b 2)s,\\
				\end{eqnarray*}
				where $\mathcal H_b=\mathcal H/\log_2 b$ and $t=t_1+t_2$. Note that $\mathcal H_b(1/2)=1/\log_2 b$.
				Now let $\Delta(T_1,T_2,r) = \delta(T_1n,T_2n,rn,n)/n$ for any $n$.
				It does not matter which $n$, since with $T=T_1+T_2$,
				$\chi(n,t,s)=c_{\mr n} n+c_{\mr t} t+c_{\mr s} s$ gives
				\[
					\frac1n\chi(n,Tn,rn)= c_{\mr n} + c_{\mr t} T + c_{\mr s} r
				\]
				and
				\[
					\frac{rn}{\chi(n,Tn,rn)}=\frac{r}{c_{\mr n} + c_{\mr t} T + c_{\mr s} r}
				\]
				and hence $\Delta(T_1,T_2,r)$ equals
				\[
				  	\sum_{i=1}^2 T_i \mathcal H_b\left(\frac12-\frac{p}{2T_i}\right)
					+ (c_{\mr n} + c_{\mr t} T 
					+ c_{\mr s} r)\mathcal H_b\left(\frac{r}{c_{\mr n} + c_{\mr t} T + c_{\mr s} r}\right)
				\]
				\[
					 + \pars{c_{\mr n} + c_{\mr t} T + c_{\mr s} r}\log_b 2
					+ (1-1/\log_2 b)r.
				\]
				\begin{lem}\label{T}
					$\Delta(T_1,T_2,r)$ is maximized at $T_1=T_2$.
				\end{lem}
				\begin{proof}
					Rewriting with $T=T_1+T_2$ and $\epsilon=T_1-T_2$,
					it suffices to show that with $g(x)=x\mathcal H(1/2-1/x)$,
					the function $f(\epsilon)=g(x+\epsilon)+g(x-\epsilon)$ is maximized at $\epsilon=0$.
					This is equivalently to $g$ being concave, which is a routine verification:
					\[
						g'(x) = \mathcal H(1/2-1/x) + x\mathcal H'(1/2-1/x)(1/x^2) = \mathcal H(1/2-1/x) + \mathcal H'(1/2-1/x)/x
					\]
					\[
						g''(x) = \mathcal H'(1/2-1/x)/x^2 + \mathcal H''(1/2-1/x)(1/x^2)(1/x) + \mathcal H'(1/2-1/x)(-1/x^2)
					\]					
					\[
						g''(x) = \mathcal H''(1/2-1/x)/x^3 < 0
					\]
				\end{proof}
				\begin{df}[Logit function]
					For any real $b>1$,
					\[
						\logit_b(x)=\log_b\left(\frac{x}{1-x}\right),\quad \logit_b:(0,1)\rightarrow \mathbb R.
					\]
				\end{df}
				A graphic of the logit function is given as Figure \ref{fig:Logit}.
				\begin{figure}[H]
					\centering
					\includegraphics[height=7cm]{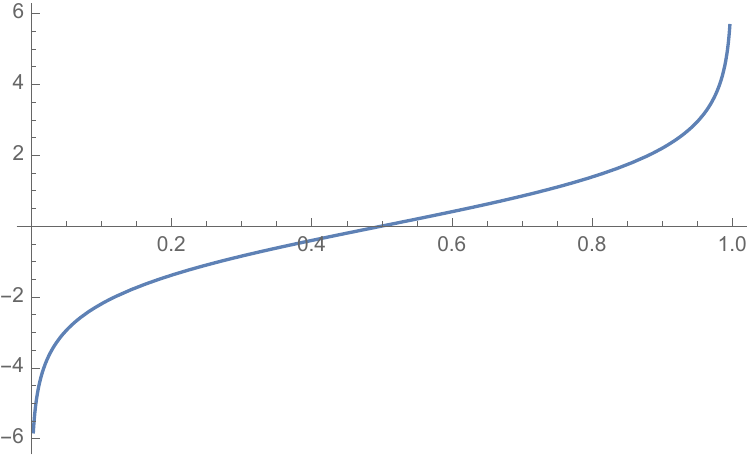}
					\caption{The logit function for $b=e$.
					Figure produced using Mathematica with $y=\log\left(\frac{x}{1-x}\right)$.}\label{fig:Logit}
				\end{figure}
				\begin{df}[Logistic sigmoid function]
					For any real $b>1$,
					\[
						\zeta_b(y)=\frac{1}{b^{-y}+1},\quad \zeta_b:\mathbb R\rightarrow (0,1).
					\]
				\end{df}
				\begin{lem}\label{trivielt}
					For any real $b>1$, the logit function $\logit_b(x)$ is a strictly increasing bijection.
					Its inverse is the logistic sigmoid function $\zeta_b(y)$.
				\end{lem}
				In light of Lemma \ref{T}, we now let $\Delta(T,r)=\Delta(T/2,T/2,r)$, so that
				\[
					\Delta(T,r) = T \mathcal H_b\left(\frac12-\frac{p}{T}\right)
					+ (c_{\mr n} + c_{\mr t} T + c_{\mr s} r)\mathcal H_b\left(\frac{r}{c_{\mr n}
					+ c_{\mr t} T + c_{\mr s} r}\right)
				\]
				\[ + (c_{\mr n} + c_{\mr t} T + c_{\mr s} r)\log_b 2 + (1-\log_b2)r.
				\]
				In the following Lemma it is useful to have the default case $(c_{\mr n},c_{\mr t},c_{\mr s})=(1,-1,0)$.
				The other case of interest is $(c_{\mr n},c_{\mr t},c_{\mr s})=(1/2,-1/2,1/2)$.
				\begin{lem}\label{rMax}
					For fixed $p$, $b$, and $T$, the function $r\mapsto \Delta(T,r)$ has a unique maximum where
					\[
						r = (c_{\mr n} + c_{\mr t} T + c_{\mr s} r)\frac{b}{b+2^{1-c_{\textrm s}}},
					\]
					hence after optimizing on $r$, $\Delta(T)-T \mathcal H_b\left(\frac12-\frac{p}{T}\right)=$
					\[
						(c_{\mr n} + c_{\mr t} T + c_{\mr s} r)
						\mathcal H_b\left(\frac{b}{b+2^{1-c_{\textrm s}}}\right)
						+ (c_{\mr n} + c_{\mr t} T + c_{\mr s} r)\log_b 2
					\]
					\[
						+ (1-\log_b2)\left[(c_{\mr n} + c_{\mr t} T + c_{\mr s} r)\frac{b}{b+2^{1-c_{\textrm s}}}\right],
					\]
					i.e.,
					\[
						\Delta(T)=T \mathcal H_b\left(\frac12-\frac{p}{T}\right)+
						(c_{\mr n} + c_{\mr t} T + c_{\mr s} r)
						\underbrace{\left\{
						\mathcal H_b\left(\frac{b}{b+2^{1-c_{\textrm s}}}\right)
						+ \log_b 2
						+ (1-\log_b2)\left[\frac{b}{b+2^{1-c_{\textrm s}}}\right]
						\right\}}_{c_b:=}
					\]
					\[
						=T \mathcal H_b\left(\frac12-\frac{p}{T}\right)+
						\frac{(c_{\mr n} + c_{\mr t} T) }{ \left(1-c_{\mr s}\frac{b}{b+2^{1-c_{\textrm s}}}\right)}
						c_b
					\]
					\[
						=T \mathcal H_b\left(\frac12-\frac{p}{T}\right)+
						(1-T)
						\underbrace{
						\frac{c_{\mr n}c_b}{\left(1-c_{\mr s}\frac{b}{b+2^{1-c_{\textrm s}}}\right)}
						}_{d_b}\quad\text{if }c_{\mr n}=-c_{\mr t}.
					\]
					namely
					\[
						r = \frac{(c_{\mr n} + c_{\mr t} T)\frac{b}{b+2^{1-c_{\textrm s}}} }{ \left(
								1-c_{\mr s}\frac{b}{b+2^{1-c_{\textrm s}}}
							\right)}
						= \frac{c_{\mr n} + c_{\mr t} T}{ \left(\frac{b+2^{1-c_{\textrm s}}}{b}-c_{\mr s}\right)}
					\]
					\[
						=\begin{cases}
							(1-T)\frac{b}{b+2}, & (c_{\mr n},c_{\mr t},c_{\mr s})=(1,-1,0) \\
							\frac12(1-T)/(\frac{b+\sqrt{2}}{b} - \frac12), & (c_{\mr n},c_{\mr t},c_{\mr s})=(1/2,-1/2,1/2)
						\end{cases}
					\]
				\end{lem}
				\begin{proof}
					We have
					\[
						\frac{d}{dx}\mathcal H_b(x)=-\logit_b(x).
					\]
					Thus by Lemma \ref{trivielt}, the inverse function of $\mathcal H_b'(x)$ is $y\mapsto \zeta_b(-y)=\frac{1}{b^y+1}$.
					Thus, we calculate
					\[
						\frac{\partial\Delta}{\partial r} = \mathcal H_b'(r/(c_{\mr n} + c_{\mr t} T
						+ c_{\mr s} r)) + 1-\log_b 2 +c_{\mr s}\log_b 2= 0\quad\text{iff}
					\]
					\[
						\frac{r}{c_{\mr n} + c_{\mr t} T + c_{\mr s} r} = (\mathcal H_b')^{-1}(\log_b 2 - 1
						- c_{\mr s}\log_b 2)
						= \frac{1}{
							b^{
								\log_b 2 - 1 - c_{\mr s}\log_b 2
							} + 1
						}
					\]
					\[	= \frac{1}{\frac{2^{1-c_{\mr s}}}b + 1}
						= \frac{b}{b+2^{1-c_{\textrm s}}},
					\]
					as desired.
					We also have
					\[
						\frac{d}{dx}\logit_b(x) = \frac{1}{\ln b} \cdot \frac{1}{x(1-x)}.
					\]
					Note that if $(c_{\mr n},c_{\mr t},c_{\mr s})=(1,-1,0)$ then $T<1$ gives $r>0$.
					Hence
					\[
						\frac{\partial^2\Delta}{\partial r^2}
						= \frac{\partial}{\partial r}\pars{-\logit_b\pars{\frac{r}{c_{\mr n} + c_{\mr t} T + c_{\mr s} r}}}
					\]
					\[
						= (-1)\frac{\partial}{\partial r}\left(\frac{r}{c_{\mr n} + c_{\mr t} T
						+ c_{\mr s} r}\right)\cdot\frac{1}{\ln(b)}\cdot\frac{1}{x(1-x)}\bigg\rvert_{x=\frac{r}{c_{\mr n}
						+ c_{\mr t} T + c_{\mr s} r}}
					\]
					\[
						=(-1)\left(\frac{c_{\mr n} + c_{\mr t} T + c_{\mr s} r - rc_{\mr s}}{(c_{\mr n}
						+ c_{\mr t} T
						+ c_{\mr s} r)^2}\right)\cdot\frac{1}{\ln(b)}\cdot\frac{1}{x(1-x)}\bigg\rvert_{x=\frac{r}{c_{\mr n}
						+ c_{\mr t} T + c_{\mr s} r}} < 0
					\]
					provided $c_{\mr t}=-c_{\mr n}<0$ and $c_{\mr s}>0$, as in our two cases.
				\end{proof}
				In light of Lemma \ref{rMax}, we now let $r=(1-T)\frac{b}{b+2}$ in $\Delta(T,r)=\Delta(T/2,T/2,r)$, giving
				\[
					\varphi(T,p) := T \mathcal H_b\left(\frac12-\frac{p}{T}\right) + (1-T)\mathcal H_b\left(\frac{r}{1-T}\right)
					+ (1-T)\log_b 2 + (1-\log_b2)r
				\]
				\[
					= T \mathcal H_b\left(\frac12-\frac{p}{T}\right) + (1-T)\mathcal H_b\left(\frac{b}{b+2}\right)
					+ (1-T)\log_b 2 + (1-\log_b2)(1-T)\frac{b}{b+2}
				\]
				which we will call $\varphi(T,p)$.
				
				To simplify calculations to come, we make Definition \ref{abbrev}.
				\begin{df}[Abbreviations]\label{abbrev}
					\begin{align*}
						T(p) &:=
							\frac{2p}{\sqrt{1      -\frac{4}{(b+2)^2}}}
						=   \frac{4p}{\sqrt{3}},\quad b=2.\\
						\varphi(T, p) &:= \Delta\left(T,(1-T)\frac{b}{b+2}\right)\\
						\beta(T) &:= \frac12-\frac{p}{T}.
					\end{align*}
				\end{df}
				Note that
				$\alpha_2 = \frac{2\cosh^{-1}(2)}{\ln 2}=2\log_2(2+\sqrt{3})$.
				Then
				\[
					\lim_{n\rightarrow\infty}\frac{\log_2 N}{n}\le \varphi(T, p) = T \mathcal H_b\left(\beta(T)\right)
			  	+(1-T)c_b.
				\]
				\begin{lem}\label{above}
					Suppose $0\le 2p\le T\le 1$ and $b\ge 2$. Suppose
					\[
						\varphi(T, p) = T \mathcal H_b\left(\beta(T) \right) + (1-T)d_b.
					\]
					for some constant $d_b$.
					Then we have
					\[
				 		0< \frac{\partial\varphi}{\partial T} \quad\Longleftrightarrow\quad	T < T(p)
					\]
					where $T(p)=\frac{2}{\sqrt{1-4b^{-2d_b}}}p$.
				\end{lem}
				\begin{proof}
					We have, using the further abbreviation $\beta=\beta(T)$,
					\[
						\frac{\partial\varphi}{\partial T}
					= \mathcal H_b\left(\beta\right)
					+ T\mathcal H_b'\left(\beta\right)\cdot \frac{p}{T^2} - d_b
					\]
					\[
						= \mathcal H_b(\beta) + \mathcal H_b'(\beta)\frac{p}{T} - d_b
						= \mathcal H_b(\beta) + \mathcal H_b'(\beta)(1/2-\beta) - d_b
					\]
					Note that $b^{-\mathcal H_b(x)}=x^x(1-x)^{1-x}$ and $b^{-\mathcal H_b'(x)}=x/(1-x)$. Thus
					now $0<\partial\varphi/\partial T$ iff $b^{-0} > b^{-\partial\varphi/\partial T}$ iff
					\[
						1 > \beta^{\beta}(1-\beta)^{(1-\beta)} \left(\frac{\beta}{1-\beta}\right)^{1/2-\beta} b^{d_b}
						= \beta^{1/2}(1-\beta)^{1/2}b^{d_b},\quad\text{iff}
					\]
					\[
						1 > \beta(1-\beta)b^{2d_b}.
					\]
					Since $0\le\beta\le 1/2$, this gives
					\begin{eqnarray*}
						\beta 						<& \frac{1-\sqrt{1-4b^{-2d_b}}}{2},&\\
						\frac{p}{T}= \frac12-\beta 	>& \frac{\sqrt{1-4b^{-2d_b}}}{2}, & \text{and}\\
						T 							<& \frac{2}{\sqrt{1-4b^{-2d_b}}}p.&
					\end{eqnarray*}
				\end{proof}
				\begin{cor}\label{below}
					Suppose $(c_{\mr n},c_{\mr t},c_{\mr s})=(1,-1,0)$.
					Suppose $0\le 2p\le T\le 1$ and $b\ge 2$. Then we have
					\[
				 		0< \frac{\partial\varphi}{\partial T} \quad\Longleftrightarrow\quad	T < T(p) := \frac{2}{\sqrt{1-4b^{-2d_b}}}p.
					\]
				\end{cor}
				\begin{proof}
					We let
					\[
						d_b = \frac{c_{\mr n}c_b}{1-c_{\mr s}\frac{b}{b+2^{1-c_{\textrm s}}}}
					\]
					\[
						=\begin{cases}
							c_b,& \text{if }(c_{\mr n},c_{\mr t},c_{\mr s})=(1,-1,0),\\
							\frac{\frac12c_b}{1-\frac12\frac{b}{b+\sqrt{2}}},
							& \text{if }(c_{\mr n},c_{\mr t},c_{\mr s})=\frac12(1,-1,1).
						\end{cases}
					\]
					If additionally we set $b=2$ then this is
					\[
						=\begin{cases}
							2,& \text{if }(c_{\mr n},c_{\mr t},c_{\mr s})=(1,-1,0),\\
							\sqrt{2}   ,& \text{if }(c_{\mr n},c_{\mr t},c_{\mr s})=\frac12(1,-1,1).
						\end{cases}
					\]
					We apply Lemma \ref{above}. Then $0<\partial\varphi/\partial T$ iff
					\[
						\beta < \frac{1-\sqrt{1-4b^{-2d_b}}}{2} =
							\frac{1-\sqrt{3/4}}{2},\quad\text{if }b=2\text{ and }(c_{\mr n},c_{\mr t},c_{\mr s})=(1,-1,0),\text{ and }\\
					\]
					\[
						\frac{p}{T}= \frac12-\beta 	> \frac{\sqrt{1-4b^{-2d_b}}}{2}=
							\frac{\sqrt{3/4}}{2},\quad\text{under the same condition.}\\
					\]
					So
					\[
						T< \frac{2}{\sqrt{1-4b^{-2d_b}}}p=
						\begin{cases}
							(2.3094)p=\frac{4}{\sqrt 3}p,					& b=2\text{ and }(c_{\mr n},c_{\mr t},c_{\mr s})=(1,-1,0)\\
							(3.0259)p=\frac{4}{\sqrt{4-4^{2-\sqrt{2}}}}p	& b=2\text{ and }(c_{\mr n},c_{\mr t},c_{\mr s})=\frac12(1,-1,1).
						\end{cases}
					\]
					Note that $b\ge 2$ and
					\[
						4b^{-2c_b}
						= \left(\frac{4}{b(b+2)}\right)^2 b^{-2(\log_b 2 - 1)}
					\]
					\[
						= \left(\frac{4}{b(b+2)}\right)^2 \frac{b^2}{4}
						= \frac{4}{(b+2)^2} \le \frac14 < 1
					\]
					give $1-4b^{-2c_b}>0$, as required.
				\end{proof}
				Let $L_b = \sqrt{1-4b^{-2d_b}}$.
				Note that $T(p)\le 1$ iff
				\[
					p\le L_b/2=
					\frac{\sqrt{1-4b^{-2d_b}}}{2}=
					\begin{cases}
						\frac{\sqrt 3}{4}=0.433,					& b=2\text{ and }(c_{\mr n},c_{\mr t},c_{\mr s})=(1,-1,0)\\
						\frac{\sqrt{4-4^{2-\sqrt{2}}}}{4}=0.330	& b=2\text{ and }(c_{\mr n},c_{\mr t},c_{\mr s})=\frac12(1,-1,1).
					\end{cases}
				\]
				and
				\[
					\varphi(T(p),p) = T(p) \mathcal H_b\left(\frac12-\frac{p}{T(p)}\right)
			  		+(1-T(p))d_b
				\]
				\[
					= \frac{2p}{L_b} \mathcal H_b\left(\frac12-\frac{L_b}{2}\right) + (1-\frac{2p}{L_b})d_b
				\]
				\[
					= d_b - \underbrace{
					\left(d_b-\mathcal H_b\left(\frac12-\frac{L_b}{2}\right)\right)\frac{2}{L_b}
					}_{\alpha_b}p
				\]
				Note
				\[
					L_2 =
					\begin{cases}
						\sqrt{3}/2					& \vec{c}=(1,-1,0)\\
						\sqrt{4-4^{2-\sqrt{2}}}/2=\sqrt{1-4^{1-\sqrt{2}}}	& \vec{c}=\frac12(1,-1,1),
					\end{cases}
				\]
				so
				\[
					\mathcal H_2\pars{\frac12-\frac{L_2}2} = 
					\begin{cases}
						0.354579,& \vec{c}=(1,-1,0)\\
						0.656615,& \vec{c}=\frac12(1,-1,1).
					\end{cases}
				\]
				Now we need
				\[
					\alpha_2 =
					\begin{cases}
						\left(2-\mathcal H_2\left(\frac12-\frac{L_2}{2}\right)\right)\frac{2}{L_2} = 2\log_2(2+\sqrt{3})=3.7999
						& \vec{c}=(1,-1,0)\\
						\left(\sqrt{2}-0.656615\right)\frac{2}{\sqrt{1-4^{1-\sqrt{2}}}}=2.29244
						& \vec{c}=\frac12(1,-1,1).
					\end{cases}
				\]
				Hence
				\[
					\lim_{n\rightarrow\infty}\frac{\log_2 N}{n}\le
					\psi(p):=\varphi\left(\min\left\{1,T(p)\right\}, p\right)
				\]
				\[
					= \begin{cases}
						\varphi(1,p)    = \mathcal H_b(1/2-p), & p\ge L_b/2;\\
						\varphi(T(p),p) = d_b - \alpha_b p,                      & p\le L_b/2.
					\end{cases}
				\]
				Note
				\[
					\psi'(p) =
					\begin{cases}
						\partial_1 \varphi(1,p)\cdot 0 + \partial_2 \varphi(1,p)\cdot 1& p<L_b/2\\
						\partial_1 \varphi(T(p),p)\cdot T'(p) + \partial_2 \varphi(T(p),p)\cdot 1& p>L_b/2.
					\end{cases}
				\]
				We can see that $\psi$ will be differentiable at the breakpoint as follows: by lemma above,
				$\partial_1\varphi(T,p)=0$ exactly at $T=T(p)$, so the first terms are both 0.
				The second terms are equal since $T(p)=1$ when $p=L_b/2$.
				That is, we apply the following lemma with $a=1$ and $L=L_b/2$.
				\begin{lem}
					Suppose $\varphi(T,p)$ is differentiable.
					Let $T(p)$ be the value of $T$ such that $\partial_1\varphi(T,p)=0$,
					let $L$ (depending on $a$) be such that for all $p$,
					\[
						T(p)\ge a\qquad\text{iff}\qquad p\le L,
					\]
					and define the function $\psi$ by
					\[
						\psi(p) = \varphi(\min(a,T(p)),p).
					\]
					Then $\psi$ is differentiable at $L$.
				\end{lem}
				\begin{proof}
					\[
						\psi'(p) =
						\begin{cases}
							\partial_1 \varphi(a,p)\cdot 0 + \partial_2 \varphi(a,p)\cdot 1& p<L\\
							\partial_1 \varphi(T(p),p)\cdot T'(p) + \partial_2 \varphi(T(p),p)\cdot 1& p>L.
						\end{cases}
					\]
				\end{proof}
				Another way is to note that $\frac{d}{dp}\varphi(1,p)=\log_2(\frac12-x)-\log_2(\frac12+x)$,
				which at $\sqrt{3}/4$ is $2\log_2(2-\sqrt{3})<0$.
				On the other hand
				$\frac{d}{dp}\varphi(T(p),p)=-\alpha_b=-2\log_2(2+\sqrt{3})=2\log_2(\frac{2-\sqrt{3}}{(2+\sqrt{3})(2-\sqrt{3})})$,
				so $\psi$ is actually differentiable at the breakpoint when $\vec c=(1,-1,0)$.
				In fact, we have differentiability for any $\vec c$ with $c_{\mr n}=-c_{\mr t}$, by the identity
				\[
					\log_b\left(\frac12-L_b/2\right)-\log_b\left(\frac12+L_b/2\right) = -\alpha_b
					= -\left(d_b-\mathcal H_b\left(\frac12-\frac{L_b}{2}\right)\right)\frac{2}{L_b}
				\]
				which follows from (and is equivalent to)
				\[
					b^{-2d_b}
					= \frac14-\frac{L_b^2}{4},
				\]
				where $L_b=\sqrt{1-4b^{-2d_b}}$.
				
				Consequently $\tilde h^*(\psi(p))\le p$. Since $\tilde h$ is decreasing it follows that $\tilde h(p)\le \psi(p)$.

			This completes the proof of Theorem \ref{best}. 
	\bibliographystyle{plain}
	\bibliography{structure-function-TCS}
\end{document}